\documentclass[12pt]{article}
\usepackage{latexsym}
\usepackage{graphicx}
\usepackage{epstopdf}
\usepackage{amsthm}
\title{Coloring directed cycles}
\author{Andrzej Szepietowski \\
         Institute of Informatics, University of Gda\'nsk \\
        ul. Wita Stwosza 57, 80952 Gda\'nsk, Poland \\
                {\tt matszp{@}inf.univ.gda.pl}}

\date{}
\begin{document}
\maketitle
\newtheorem{df}{Definition}
\newtheorem{lem}[df]{Lemma}
\newtheorem{thm}[df]{Theorem}
\newtheorem{cor}[df]{Corollary}

\begin{abstract}
Sopena in his survey 
[E. Sopena, 
The oriented chromatic number of graphs: A short survey, 
preprint 
2013]
writes, without any proof, that an oriented cycle $\vec C$ can be colored with  three colors if and only if
$\lambda(\vec C)=0$, where  $\lambda(\vec C)$ is the number of forward arcs minus the number of backward arcs in $\vec C$.
This is not true. In this paper we show that
$\vec C$ can be colored with three colors if and only if 
$\lambda(\vec C)=0~(\bmod~3)$ or $\vec C$
does not contain three consecutive arcs going in the same direction.
\end{abstract}

\section{Introduction}
Oriented coloring is a  coloring of the vertices of an  oriented graph
$G$ such that: (1) no two neighbors have the same color, (2)
for every two arcs $(t,u)$ and $(v,w)$, either $\beta(t)\ne \beta(w)$ or
$\beta(u)\ne \beta(v)$.
In other words, if there is an arc leading  from the color $\beta_1$ to
$\beta_2$, then no arc leads from $\beta_2$ to $\beta_1$.
The oriented chromatic number $\chi_o({G})$ of an oriented graph $G$  is the
smallest number $k$ of colors needed to color $G$.
It is easy to see that an oriented graph can be colored by $k$ colors if
and only if there exists an homomorphism from $G$ to an oriented
graph $H$ with $k$ vertices.

Oriented coloring has been studied  in [1--11], 
see~\cite{sop2} for a short survey of the main results.
An oriented cycle of length $n\geq 3$ is a sequence of vertices
$\vec C=v_0,v_1,\dots,v_{n-1}$, each vertex $v_i$, $0\le i\le n-1$, is connected with the next vertex by
an arc, either by $(v_i,v_{i+1})\in A(\vec{C})$ (forward arc) or by $(v_{i+1},v_i)\in A(\vec{C})$
(backward arc); $v_{n}$ denotes $v_0$. For each pair of consecutive vertices $v_i$ and $v_{i+1}$, we define
$\lambda(v_i,v_{i+1})=1$ if $(v_i,v_{i+1})\in A(\vec{C})$, and
$\lambda(v_i,v_{i+1})=-1$ if $(v_{i+1},v_i)\in A(\vec{C})$.
For the whole cycle
$$\lambda(\vec C)=\sum_{i=0}^{n-1}\lambda(v_i,v_{i+1}).$$
In other words $\lambda(\vec C)$ is the number of forward arcs minus the number of backward arcs in $\vec C$. In his paper~\cite{sop2} Sopena gives, without any proof, the following lemma:

\begin{lem}[\cite{sop2}] If $\vec C$ is an oriented cycle then:
\begin{enumerate}
\item $\chi_o({\vec C})=2$ if and only if the orientation of $\vec C$ is alternating.
\item $\chi_o({\vec C})=3$ if and only if $\lambda(\vec C)=0$ and the orientation of $\vec C$ is not alternating.
\item $\chi_o({\vec C})=4$ if and only if $\lambda(\vec C)\ne0$ and $\vec C$ is not the directed cycle 
$\vec C_5$ on 5 vertices (see Fig.~\ref{c5}).
\end{enumerate}
\end{lem}

\begin{figure}[htb]
\centering
\includegraphics[scale=0.7]{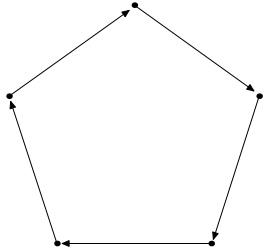}
\caption{Oriented cycle $\vec C_5$}
\label{c5}
\end{figure}

It is easy to see that 1. is correct.
In the cycle ${\vec C}_5$ with all five arcs going in one direction (see Fig.~\ref{c5}), every two vertices can be linked by 
a directed path  of length 1 or 2, therefore $\chi_o({\vec C}_5)=5$.
It is also true that every directed cycle $\vec C$ except
${\vec C}_5$ can be colored with four colors. We give a proof
of this in Section~\ref{s4}. But the characterization of cycles
colored with three colors given in Lemma~1 is not correct. In this paper we shall prove the following:
\begin{lem} If $\vec C$ is an oriented cycle then:
\begin{enumerate}
\item $\chi_o({\vec C})=2$ if and only if the orientation of $\vec C$ is alternating.
\item $\chi_o({\vec C})=3$ if and only if $\lambda(\vec C)=0~(\bmod~3)$ or $\vec C$ does not contain three consecutive arcs going in the same direction, and the orientation of $\vec C$ is not alternating.
\item $\chi_o({\vec C})=5$ if and only if $\vec C = \vec C_5$.
\item $\chi_o({\vec C})=4$ otherwise. 
\end{enumerate}
\end{lem}

In Section~\ref{s3} we describe oriented cycles colored with three colors
and in Section ~\ref{s4} we shall prove that every directed cycle $\vec C$ except
${\vec C}_5$ can be colored with four colors.

\section{Coloring with three colors}\label{s3}
For every coloring oriented graph $G$ with three vertices, there is a homomorphism from $G$
either into the cycle ${\vec C}_3$, see Fig.~\ref{c3}, or to the cycle ${\vec V}_3$, see Fig.~\ref{v3}. Thus
if a cycle $\vec C$ is colored with three colors, then there is a homomorphism from $\vec C$ 
into ${\vec C}_3$ or ${\vec V}_3$.

\begin{figure}[htb]
\centering
\includegraphics[scale=0.7]{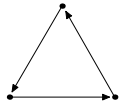}
\caption{Oriented cycle $\vec C_3$}
\label{c3}
\end{figure}
\begin{figure}[htb]
\centering
\includegraphics[scale=0.7]{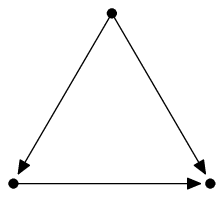}
\caption{Oriented cycle $\vec V_3$}
\label{v3}
\end{figure}

\begin{lem}\label{l2} 
An oriented cycle $\vec C$ can be colored by a homomorphism into ${\vec C}_3$ if and only if $\lambda(\vec C)=0~(\bmod~3)$.
\end{lem}
\begin{proof}
Suppose that the cycle $\vec C=v_0,v_1,\dots,v_{n-1}$ and $\lambda(\vec C)=0~(\bmod~3)$.
Consider now the function $h:\vec C\to {\vec C}_3$ defined by
$h(v_0)=0$ and
for every $1\le k\le n-1$,
$$h(v_k)=\sum_{i=0}^{k-1}\lambda(v_i,v_{i+1})~(\bmod~3)$$
It is easy to see that $h$ is a coloring homomorphism.

Suppose now that there is a homomorphism $h:\vec C\to {\vec C}_3$.
Then it is easy to see that, 
$h(v_1)=h(v_0)+\lambda(v_0,v_{1})~(\bmod~3)$
and
for every $1\le k\le n-1$,
$$h(v_k)=h(v_0)+\sum_{i=0}^{k-1}\lambda(v_i,v_{i+1})~(\bmod~3).$$
Finally
$$h(v_0)=h(v_0)+\sum_{i=0}^{n-1}\lambda(v_i,v_{i+1})~(\bmod~3).$$
Thus,
$$\lambda(\vec C)=\sum_{i=0}^{n-1}\lambda(v_i,v_{i+1})=0~(\bmod~3).$$
\end{proof}
\begin{lem}
An oriented cycle $\vec C$ can be colored by a homomorphism into ${\vec V}_3$ if and only if $\vec C$
does not contain three consecutive arcs going in the same direction ($\to\to\to$ or $\gets\gets\gets$).
\end{lem}
\begin{proof}
It is easy to see that if $\vec C$ can be colored by a homomorphism into ${\vec V}_3$,
then $\vec C$ cannot contain three consecutive arcs going in the same direction.

On the other hand if $\vec C$
does not contain three consecutive arcs going in the same direction, then define the 
function $h:\vec C\to {\vec V}_3$ to be $h(v)=indegree(v)$.
To see that $h$ is a homomorphism, consider two neighbors $v_i$ and $v_{i+1}$, and suppose that 
there is arc  from $v_i$ to $v_{i+1}$.
Then:
\begin{itemize}
\item $h(v_i)\ne 2$, because $indegree(v_i)\le 1$, 
\item$h(v_{i+1})\ne 0$,  because $indegree(v_{i+1})\ge 1$,
\item $h(v_i)=h(v_{i+1})=1$ is not possible, because otherwise three arcs
$(v_{i-1},v_i)$, 
$(v_{i},v_{i+1})$, 
$(v_{i+1},v_{i+2})$, will have the same direction.
\end{itemize}
\end{proof}

Thus we have the following Lemma.

\begin{lem}
An oriented cycle $\vec C$ can be colored with three colors if and only if 
$\lambda(\vec C)=0~(\bmod~3)$ or $\vec C$
does not contain three consecutive arcs going in the same direction ($\to\to\to$ or $\gets\gets\gets$).
\end{lem}

\section{Coloring with four colors}\label{s4}

The following lemma is in~\cite{sop2} without any proof.

\begin{lem}
Every oriented cycle $\vec C$ except ${\vec C}_5$ can be colored with 4 colors.
\end{lem}
\begin{proof}
Every cycle with 3 or 4 vertices can be colored with 4 colors.

Suppose first that $\vec C$ of length $n\ge 6$ has all $n$ arcs going in one direction.
Then there are two nonnegative integers $a$, $b$, such that $n=3a+4b$, and 
$\vec C$ can be colored by a homomorphism into the graph with arcs
$0\to 1\to 2\to3 \to 0$ and $2\to 0$.

If $\vec C$ has arcs in opposite directions, then it has a vertex, say $v_{n-1}$, of outdegree 0.
In this case color $v_{n-1}$ with color 3 and the path $v_0, v_1,\dots, v_{n-2}$ with colors
$0,1,2$ as in Lemma~\ref{l2}
\end{proof}

\end{document}